\documentclass[letterpaper,11pt]{article}

\usepackage{amsthm}
\usepackage{amssymb}
\usepackage{amsmath}
\usepackage{graphicx}
\usepackage{algorithm}
\usepackage{algorithmic}
\usepackage{fullpage}

\newtheorem{thm}{Theorem}[section]

\newtheorem{lem}[thm]{Lemma}
\newtheorem{obs}[thm]{Observation}
\newtheorem{prop}[thm]{Proposition}
\newtheorem{assmp}[thm]{Assumption}
\theoremstyle{definition}
\newtheorem{defn}[thm]{Definition}
\theoremstyle{remark}

\numberwithin{equation}{section}

\newcommand{\R}{{\mathbb R}}
\newcommand{\eps}{\varepsilon}

\def\D{{\mathcal D}}
\def\S{{\mathcal S}}
\def\dtau{d_{\tau}}
\def\dfoot{d_{\operatorname{foot}}}
\def\polylog{\operatorname{polylog}}

\def\one{{\bf 1}}

\def\localimprove{\operatorname{ApproxLocalImprove}}
\def\samplerank{\operatorname{SampleAndDecompose}}
\def\samplerankroot{\operatorname{SampleAndRank}}
\def\ES{{\mathcal E}_1}
\def\Ecostapprox{{\mathcal E}_2}

\def\myphi{\varphi}
\def\negspaceA{} 
\def\negspaceB{} 
\def\negspaceC{} 
\def\negspaceD{} 
\def\negspaceeq{} 
\def\neghspace{} 
\def\score{\operatorname{score}}

\def\internal{\mathcal I}
\def\leaf{\mathcal L}

\def\root{{\operatorname{RT}}}
\newcommand{\leftc}[1]{L[#1]}
\newcommand{\rightc}[1]{R[#1]}
\def\big{{\mathcal B}}

\begin{document}
\title{An Active Learning Algorithm for Ranking from Pairwise Preferences with an Almost Optimal Query Complexity}
\author{Nir Ailon\thanks{Technion, \texttt{nailon@cs.technion.ac.il}}}
\setcounter{page}{0}
\maketitle
\thispagestyle{empty}
\maketitle

\begin{abstract}
We study the problem of learning to rank from pairwise preferences, and solve a long-standing open problem that has led
to development of many heuristics but no provable results for our particular problem.

The setting is as follows:  We are given  a set $V$ of  $n$ elements from some universe, 
and we wish to linearly order them given pairwise preference labels.  given two elements $u,v\in V$, a  pairwise preference label
is obtained as a response, typically from a human, to the question \emph{which if preferred, $u$ or $v$?}
We assume no abstention, hence, either $u$ is preferred to $v$ (denoted $u \prec v$) or the other way around.
We also assume possible non-transitivity paradoxes which may arise naturally due to human mistakes or irrationality.

The goal is to linearly order the elements from the most preferred to the least preferred, while disagreeing with
as few pairwise preference labels as possible.  Our performance is measured by two parameters:  The loss (number
of pairwise preference labes we disagree with) and the query complexity (number of pairwise preference labels we
obtain).  This is a typical learning problem, with the exception that the space from which the pairwise preferences
is drawn is finite, consisting of ${n\choose 2}$ possibilities only. 
Our algorithm reduces this problem to another  problem, for which any standard learning black-box
can be used.  The advantage of the reduced problem compared to the original one is the fact that
never more than $O(n\polylog(n,\eps^{-1}))$ labels are needed (including the query complexity  of the reduction) in order 
to obtain the same risk that the same black-box
would have incurred given access to all possible ${n\choose 2}$ labels in the original problem, up to a multiplicative regret of $(1+\eps)$.
The label sampling is adapative, hence, viewing our algorithm as a preconditioner for a learning black-box
we arrive at an \emph{active learning algorithm with provable, almost optimal bounds}.  We also show that VC arguments
 give significantly worse query complexity bounds
for the same regret in a non-adaptive sampling strategy.

Our main result settles an open problem posed by 
learning-to-rank theoreticians and practitioners:
What is a provably correct way to sample preference labels?

To further  show the power and practicality of our solution, we analyze a typical test case in which the learning black-box preconditioned by our algorithm is a regularized large margin linear classifier.  

\end{abstract}

\newpage
\negspaceB
\section{Introduction}
\negspaceA
We study the problem of learning to rank from pairwise preferences, and solve a long-standing open problem that has led
to development of many heuristics but no provable results.

The setting is as follows:  We are given  a set $V$ of  $n$ elements from some universe, 
and we wish to linearly order them given pairwise preference labels.  given two elements $u,v\in V$, a  pairwise preference label
is obtained as a response, typically from a human, to the question \emph{which if preferred, $u$ or $v$?}
We assume no abstention, hence, either $u$ is preferred to $v$ (denoted $u \prec v$) or the other way around.

The goal is to linearly order the elements from the most preferred to the least preferred, while disagreeing with
as few pairwise preference labels as possible.  Our performance is measured by two parameters:  The loss (number
of pairwise preference labes we disagree with) and the query complexity (number of pairwise preference labels we
obtain).  This is a typical learning problem, with the exception that the sample space is finite, consisting of ${n\choose 2}$ possibilities only. 

The loss minimization problem given the entire $n\times n$ preference matrix is a well known NP-hard
problem called MFAST (minimum feedback arc-set in tournaments) \cite{AlonRankingTournaments}.
Recently, Kenyon and Schudy  \cite{DBLP:conf/stoc/Kenyon-MathieuS07} have devised a PTAS for it, namely,
a polynomial (in $n$) -time algorithm computing a solution with loss at most $(1+\eps)$ the optimal, for and $\eps>0$ (the
degree of the polynomial may depend on $\eps$).
In our case
each edge from the input graph is given for a unit cost.  
%
Our main algorithm is derived from Kenyon et al's algorithm.  Our output, however, is not a solution
to MFAST, but rather a reduction of the original learning problem to a different, simpler one.  The reduced
problem can be solved using any general ERM (empirical risk minimization) black-box.  
The sampling of preference labels from the original problem is adaptive, hence the combination of our algorithm and any ERM blackbox is an active learning one.
We give examples with an SVM based ERM black-box toward the end, and show that our approach gives rise to a reduced SVM problem which provably approximates the original problem to within any arbitrarily small
error \emph{relative to the original SVM optimal solution}.  The total number of pairwise preference labels acquired in the reduction and in the construction of the reduced SVM is significantly smaller than what a VC-dimension type argument would guarantee.

Our setting defers  from much of the  \emph{learning to rank} (LTR) literature.  Usually, the labels used in 
LTR problems are responses to individual elements, and not to pairs of elements. A typical example is
the $1..5$ scale rating for restaurants, or $0,1$ rating (irrelevant/relevant) for candidate documents retrieved
for a query (known as the \emph{binary ranking} problem).  The goal there is, as in ours, to order the elements while disagreeing with as little pairwise
relations as possible, where a pairwise relation is derived from any two elements rated differently.  Note that the underlying preference graph there is transitive, hence no combinatorial problem due to nontransitivity.  In fact, some view the rating setting as
an ordinal regression problem and not a ranking problem.  
Here the preference graph may contain cycles, and is hence agnostic with respect to the concept class 
we are allowed to output from, namely, permutations.  We note that some LTR literature does consider the pairwise preference
label approach, and there is much justification to it (see \cite{Carterette08hereor,1414238} and reference therein).
As far as we know, our work provides 
 a sound solution to a problem addressed by machine learning practitioners  (e.g. \cite{Carterette08hereor}) who use pairwise preferences as labels
for the task of learning to rank items, but wish to avoid obtaining labels for the  quadratically many preference pairs, without compromising low error bounds.  We also show that the \emph{fear of quadraticity} found in much work dealing with
pairwise preference based learning to rank (e.g., from Crammer et. al \cite{Crammer01prankingwith} \emph{the [pairwise] approach is time consuming since it requires
increasing the sample size ... to $O(n^2)$}) is unfounded in the light of new advances in combinatorial optimization \cite{DBLP:journals/jacm/AilonCN08,DBLP:conf/stoc/Kenyon-MathieuS07}.


%

It is important to note a significant difference between our work and   Kenyon and Schudy's PTAS \cite{DBLP:conf/stoc/Kenyon-MathieuS07}, which is also
the difference between the combinatorial optimization problem and the learning counterpart.
A good way to explain this  to to compare two learners, Alice and Bob.  On the first day, Bob queries all ${n\choose 2}$ pairwise preference 
labels and sends them to a perfect solver for MFAST. Alice uses our work to query only $O(n\polylog(n,\eps^{-1}))$ preference labels 
amd obtains a decomposition of the original input $V$ into an \emph{ordered list} of sub-problems $V_1,\dots, V_k$ where
each $V_i$ is contained in $V$.  Using the same  optimizer for each part and concatenating the individual output permutations, Alice will incur a loss of at most $(1+\eps)$
that of Bob.  So far Alice might not gain much, because the decomposition may consist of a single block, hence no reduction. The next day, Bob realizes that his MFAST solver cannot deal with large inputs because he is trying to solve an NP-Hard problem.  Also, he seeks a multiplicative regret of $(1+\eps)$ with respect to the optimal solution (we also say a \emph{relative regret} of $\eps$), and  his sought  $\eps$ is too small to use the PTAS.\footnote{The running time of the PTAS is exponential in $\eps^{-1}$.} 
  To remedy this, he takes advantage of the fact that the set $V$ does not merely consist of abstract elements, but rather each $u\in V$ is endowed with a feature vector $\myphi(u)$
and hence each pair of points $u,v$ is endowed with the combined feature vector $(\myphi(u), \myphi(v))$. As in typical
learning, he posits that the order relation between $u,v$ can be deduced from a linear function of $(\myphi(u), \myphi(v))$,
and invokes an optimizer (e.g. SVM) on the relaxed problem, with all pairs as input.  Note that  Bob may try to sample
pairs uniformly to reduce the query complexity (and, perhaps, the running time of the relaxed solver), but as we show below, he will
be discouraged from doing so because in certain realistic cases a relative regret of $\eps$ may entail sampling the entire
pairwise preference space.
Alice uses the same relaxed optimizer, say, SVM.  The labels she sends to the solver consist of a uniform sample of pairs
from each block $V_i$, together with all pairs $u,v$ residing in separate blocks from her aforementioned construction decomposition.  From the former label type
she would need only $O(n\polylog(n,\eps^{-1}))$ many, because (per our decomposition design)  within the blocks the cost of any solution
is high, and hence a  \emph{relative} error is tantamount to an absolute error of similar magnitude, for which
simple VC bounds allow low query complexity.  From the latter label type, she would generate a label for all pairs $u,v$
in distinct $V_i, V_j$, using a "made up" label corresponding to the order of $V_i, V_j$ (recall that the decomposition is ordered).  
Since both Bob and Alice used SVM with the same feature vectors (and  the same regularization), there
is no reason to believe that the additional cost incurred by the relaxation inaccuracies would hurt neither  Bob nor Alice
more than the other. The same statement applies to any relaxation (e.g. decision trees), though we will make 
a quantitative statement for the case of large margin linear classifiers below.

Among other changes to Kenyon and Schudy's algorithm, a key technique is to convert a  highly sensitive greedy improvement step  into a robust approximate one, by
careful sampling.  The main difficulty stems from the fact that after a single greedy improvement step, the sample becomes stale and requires refereshing.
We show a query efficient refreshing technique that allows iterated approximate greedy improvement steps.  
Interestingly, their original analysis is amenable to this change.
It is also interesting to note that the sampling scheme used for identifying greedy improvement steps for a current solution are similar
to ideas used by Ailon et. al \cite{DBLP:journals/rsa/AilonCCL07,DBLP:journals/algorithmica/AilonCCL08} and Halevy et. al \cite{DBLP:journals/siamcomp/HalevyK07} in the context
of property testing and reconstruction, where elements are sampled from exponentially growing intervals in a linear order.

Ailon et. al's $3$-approximation algorithm for MFAST using QuickSort \cite{DBLP:journals/jacm/AilonCN08} is used in
Kenyon et. al \cite{DBLP:conf/stoc/Kenyon-MathieuS07} as well as here as an initialization step.
Note that this is a sublinear algorithm.
In fact, it samples only $O(n\log n)$ pairs from the ${n\choose 2}$ possible, on expectation.    Note also that the
pairs from which we query the preference relation in QuickSort
are chosen adaptively.

\negspaceB
\section{Notation and Basic Lemmata}
\negspaceA
\subsection{The Learning Theoretical Problem}\label{sec:learning}
Let $V$ denote a finite set that we wish to rank.  In a more
general setting we are given a sequence $V^1, V^2, \dots$ of sets, but there is enough structure and interest in the single
set case, which we focus on in this work.
Denote by $n$ the  cardinality of $V$.
We assume there is an underlying preference function $W$ on pairs of elements in $V$, which
is unknown to us.  For any ordered pair $u,v\in V$, the preference value $W(u,v)$ takes the value of $1$ if
$u$ is deemed preferred over $v$, and $0$ otherwise.  
We enforce $W(u,v)+W(v,u)=1$, 
hence, $(V,W)$ is a tournament.  We assume that $W$ is \emph{agnostic} in the sense
that it does not necessarily encode a transitive preference function, and may contain errors and inconsistencies.
For convenience, for any two real numbers
$a,b$ we will let $[a,b]$ denote the interval $\{x: a \leq x \leq b\}$ if $a\leq b$ and $\{x: b \leq x \leq a\}$ otherwise.

\def\H{{\mathcal H}}
Assume now that we wish to predict $W$ using a hypothesis $h$ from some concept class $\H$.
The hypothesis $h$ will take an ordered pair $(u,v)\in V$ as input, and will output  label 
of  $1$ to assert that \emph{$u$ precedes $v$} and $0$ otherwise.  We want $\H$
to contain only consistent hypotheses, satisfying transitivity (i.e. if $h(u,v)=h(v,w)=1$ then $h(u,w)=1$).
A typical way to do this is using a linear score function:  Each $u\in V$ is endowed with a feature vector
$\myphi(u)$ in some RKHS $H$,  a weight vector $w\in H$ is used for parametrizing each $h_w\in \H$, and the
prediction is as follows:\footnote{We assume that $V$is endowed with an arbitrary linear order relation, so we can formally write $u<v$ to  arbitrarily yet consistently break ties.}
$$
 h_w(u,v) = \begin{cases} 1 &  \langle w, \myphi(u)\rangle > \langle w, \myphi(v) \rangle \\
                         0 &  \langle w, \myphi(u)\rangle < \langle w, \myphi(v) \rangle \\
			  \one_{u<v} & \mbox{otherwise}
                         \end{cases}\ .
$$
Our work is relevant, however, to nonlinear hypothesis classes as well.  We denote by $\Pi(V)$ the set permutations
on the set $V$, hence we always assume $\H\subseteq \Pi(V)$. (Permutations $\pi$ are naturally viewed as binary
classifiers of pairs of elements via the preference predicate:  The notation is, $\pi(u,v)=1$ if and only if $u \prec_\pi v$, namely,
if $u$ precedes $v$ in $\pi$.  Slightly abusing notation, we also view permutations as injective functions from $[n]$ to $V$, so that the element
 $\pi(1)\in V$ is in the first, most preferred position and $\pi(n)$ is the least preferred one. 
\def\rank{\rho}
We also define the function $\rank_\pi$ inverse to $\pi$ as the unique function satisfying $\pi(\rank_\pi(v)) = v$ for all $v\in V$.
Hence,  $u\prec_\pi v$ is equivalent to  $\rank_\pi(u) < \rank_\pi(v)$. )

As in standard ERM setting, we assume a non-negative risk function $C_{u,v}$ penalizing 
the error of $h$ with respect to the pair $u,v$, namely,
$$ C_{u,v}(h, V, W) = \one_{h(u,v)\neq W(u,v)}\ .$$
The total loss, $C(h, V,W)$ is defined as $C_{u,v}$ summed over all unordered $u,v\in V$.
Our goal is to devise an active learning algorithm for the purpose of minimizing this loss.  

In this paper we find an almost optimal solution to the problem using important breakthroughs in 
combinatorial optimization of a related problem called \emph{minimum feedback arc-set in tournaments} (MFAST).
The relation between this NP-Hard problem and our learning problem  has been noted before \cite{CohenLearningToOrder}, but no provable almost optimal active learning has been devised, as far as we know.

\subsection{The Combinatorial Optimization Counterpart}
MFAST is defined as follows:  Assume we are given $V$ and $W$ and its entirety, in other words, we pay
no price for reading $W$.  The goal is to  order the elemtns of $V$ in a full linear order, 
while minimizing the total pairwise violation.  
More precisely, we wish to find a permutation $\pi$ on the elements of $V$ such that the total
backward cost:
\begin{equation}\label{eq:MFASTcost} C(\pi,V,W) =  \sum_{u \prec_\pi v} W(v,u) \ \end{equation}
is minimized.  
 The expression in (\ref{eq:MFASTcost}) will be referred to as the \emph{MFAST cost} henceforth.

When $W$ is given as input, this problem is known as the minimum feedback arc-set  in tournaments (MFAST).
A PTAS has been discovered for this NP-Hard very recently \cite{DBLP:conf/stoc/Kenyon-MathieuS07}.
Though a major theoretical achievement from a combinatorial optimization point of view, the PTAS is not useful
for the purpose of \emph{learning to rank from pairwise preferences} because it is not query efficient.  Indeed,
it may require in some cases to read all quadratically many entries in $W$.
In this work we fix this drawback, while using their main ideas for the purpose of machine learning to rank.
We are not interested in MFAST per se, but use the algorithm in \cite{DBLP:conf/stoc/Kenyon-MathieuS07}
to obtain a certain useful decomposition of the input $(V,W)$ from which our main active learning result easily follows.

\begin{defn}\label{def:goodpart}
Given a set $V$ of size $n$, an ordered decomposition is a list of pairwise disjoint subsets $V_1,\dots, V_k\subseteq V$ such
that $\cup_{i=1}^k V_i = V$.  For a given decomposition, we let $W|_{V_i}$ denote the restriction of $W$ to $V_i\times V_i$ for $i=1,\dots, k$. 
Similarly, for a permutation $\pi\in \Pi(v)$ we let $\pi|_{V_i}$ denote the restriction of the permutation to the elements of $V_i$ (hence, $\pi|_{V_i} \in \Pi(V_i)$).   We say that $\pi\in\Pi(V)$ respects $V_1,\dots, V_k$ if for all $u\in V_i, v\in V_j, i<j$, $u \prec_\pi v$.  We denote the set of
permutations $\pi\in\Pi(V)$ respecting the decomposition $V_1,\dots, V_k$ by $\Pi(V_1,\dots, V_k)$.
 We say that a subset $U$ of $V$
is \emph{small in $V$} if $|U| \leq\log n/\log\log n$, otherwise we say that \emph{$U$ is big in $V$}.
A decomposition $V_1,\dots, V_k$ is $\eps$-\emph{good} 
with respect to $W$ if:\footnote{We will just say $\eps$-good if
$W$ is clear from the context.}
\begin{itemize}
\item  Local chaos: 
\begin{equation}\label{eq:localchaos} \min_{\pi \in \Pi(V)}{\sum_{i: V_i \mbox{ big in } V} C(\pi_{|V_i}, V_i, W_{|V_i})} \geq \eps^2\sum_{i: V_i\mbox{ big in } V} {n_i \choose 2} \ .\end{equation}
\item Approximate optimality:  
\begin{equation}\label{eq:globalorder} \min_{\sigma\in\Pi(V_1,\dots, V_k)}C(\sigma, V, W) \leq (1+\eps) \min_{\pi \in \Pi(V)} C(\pi, V, W)\ . 
\end{equation}
\end{itemize}
\end{defn}
Intuitively, an $\eps$-good decomposition 
identifies a block-ranking of the data that is difficult to rank in accordance with $W$ internally on average among big blocks (\emph{local chaos}),
yet possible to rank almost optimally while respecting the decomposition (\emph{approximate optimality}).
We show how to take advantage of an $\eps$-good decomposition for learning
in Section~\ref{sec:statisticallearningtheory}. 
 The ultimate goal will be to find an $\eps$-good decomposition of
the input set $V$ using $O(\polylog(n, \eps^{-1}))$ queries into $W$.  

\negspaceB
\subsection{Basic Results from Statistical Learning Theory}\label{sec:statisticallearningtheory}
\negspaceA
In statistical learning theory, one seeks to find a classifier minimizing an expected cost incurred on a random input
by minimizing the empirical cost on a sample thereof.  If we view pairs of elements in $V$ as data points, then the MFAST
cost can be cast, up to normalization, as an expected cost over a random draw of a data point.
 Recall our notation of $\pi(u,v)$ denoting the indicator function for the predicate $u\prec_\pi v$.
  Thus $\pi$ is viewed as a binary hypothesis function over ${V\choose 2}$,  and $\Pi(V)$ can be viewed as the  concept class
 of all binary hypotheses satisfying transitivity: $\pi(u,v) + \pi(v,y) \geq \pi(u,y)$ for all $u,v,y$.

\noindent
A sample $E$ of unordered pairs gives rise to a \emph{partial cost}, $C_E$ defined as follows:
\begin{defn}
 Let $(V,E)$ denote an undirected graph over $V$, which may contain parallel edges ($E$ is a multi-set).   The partial MFAST cost $C_E(\pi)$ is defined as
\negspaceeq
$$ C_E(\pi,V,W) = {n\choose 2}{|E|}^{-1}\sum_{\stackrel{(u,v)\in E}{u <_\pi v}} W(v,u)\ .$$
\end{defn}
\noindent
(The accounting of parallel edges in $E$ is clear.)  The function $C_E(\cdot, \cdot, \cdot)$
can be viewed as an \emph{empirical unbiased estimator} of $C(\pi, V, W)$ if $E \subseteq {V\choose 2}$ is chosen uniformly at random among
all (multi)subsets of a given size.

 The basic question in statistical learning theory is,
how good is the minimizer $\pi$ of $C_E$, in  terms of $C$?  The notion of  VC dimension \cite{vapnik:264} gives us a nontrivial bound which is, albeit suboptimal (as we shall soon see),  a good start for our purpose.
\begin{lem}\label{lem:linvcbound}
 The VC dimension of the set of permutations on $V$, viewed as binary classifiers on pairs of elements, is $n-1$.
\end{lem}
It is easy to show that the VC dimension is at most $O(n\log n)$.  Indeed, the number of permutations is at most $n!$, and the VC
dimension is always bounded by the log of the concept class cardinality.  That the bound is linear was proven in \cite{AilonR11}. We present the proof
here in Appendix~\ref{appendix:vcbound} for completeness.
The implications of the VC bound are as follows.
\begin{prop}\label{prop:vc}
 Assume $E$ is chosen uniformly at random (with repetitions) as a sample of $m$ elements from ${V\choose 2}$, where $m > n$.
Then with  probability at least $1-\delta$ over the sample, all permutations $\pi$ satisfy:
\begin{equation*} \left |C_E(\pi,V,W) - C(\pi,V,W) \right| =  n^2 O\left (\sqrt {\frac{n\log m + \log(1/\delta)}{m}}\right )\ .
\end{equation*}
\end{prop}

The consequence of Proposition~\ref{prop:vc} are as follows:  If we want to minimize $C(\pi, V, W)$ over $\pi$ to within an
additive error of $\mu n^2$, and succeed in doing so with probability at least $1-\delta$, it is enough to  choose a sample $E$ of $O(\mu^{-2}(n\log n+\log\delta^{-1}))$  elements from ${V\choose 2}$ uniformly at random (with repetitions),
and optimize $C_E(\pi, V, W)$.  Assume from now on
that $\delta$ is at least $e^{-n}$, so 
that we get a more manageable sample  bound of $O(\mu^{-2}n\log n)$.
Before turning to optimizing $C_E(\pi,V, W)$, a hard
problem in its own right  \cite{Karp, DS},
we should first understand whether this bound is at all good for various scenarios.
We need some basic notions of distance between permutations.
For two permutations $\pi, \sigma$, the Kendall-Tau  distance $\dtau(\pi, \sigma)$ is defined as
\negspaceeq
$$ \dtau(\pi, \sigma) = \sum_{u\neq v} \one[(u\prec_\pi v) \wedge (v \prec_\pi u)]\ .$$
The Spearman Footrule distance $\dfoot(\pi,\sigma)$ is defined as
\negspaceeq
$$ \dfoot(\pi, \sigma) = \sum_u |\rank_\pi(u) - \rank_\sigma(u)|\ .$$
The following is a well known inequality due to Graham and Diaconis \cite{diaconis} relating the two distance measures for all $\pi,\sigma$:
\negspaceeq
\begin{equation}\label{eq:diaconis}
 \dtau(\pi, \sigma) \leq \dfoot(\pi, \sigma) \leq 2\dtau(\pi, \sigma)\ .
 \end{equation}

Clearly $d_\tau$ and $\dfoot$ are metrics.  It is also clear that $C(\cdot, V, \cdot)$ is an extension of $d_\tau(\cdot, \cdot)$ to distances between
permutations and binary tournaments, with the triangle inequality of the form $d_\tau(\pi, \sigma) \leq C(\pi, V, W) + C(\sigma, V, W)$ satisfied
for all $W$ and $\pi, \sigma \in \Pi(V)$.

Assume now that we are able, using Proposition~\ref{prop:vc} and the ensuing comment, to find a solution $\pi$ for MFAST, with an additive regret of $O(\mu n^2)$ with respect to an optimal solution $\pi^*$ for some $\mu>0$.  
The triangle inequality implies that the distance $d_\tau(\pi, \pi^*)$ between our solution and the true optimal is $\Omega(\mu n^2)$.  By (\ref{eq:diaconis}),
this means that $\dfoot(\pi, \pi^*) = \Omega(\mu n^2)$.  By the definition of $\dfoot$, this means that the  averege element $v\in V$ is translated $\Omega(\mu n)$ positions away from its position in $\pi^*$.  In a real life application (e.g. in information retrieval), one may want elements to be at most a constant $\gamma$ positions away from their position in a correct permutation.  This translates to a sought regret of $ O(\gamma n )$ in $C(\pi, V, W)$, or, using the above notation, to $\mu = \gamma/n$.  Clearly,
Proposition~\ref{prop:vc} cannot guarantee less than a quadratic sample size for such a regret, 
which is tantamount to querying $W$
in its entirety.  {\emph We can do better:}  In this work, for any $\eps>0$  we will achieve a regret of  $O(\eps C(\pi^*, V, W))$ using  $O(\polylog(n, \eps^{-1}))$ queries into $W$, regardless of how small the optimal cost $C(\pi^*, V, W)$ is.  Hence,
our regret is relative to the optimal loss.
This is clearly not achievable using Proposition~\ref{prop:vc}.

\noindent
Before continuing, we need need a slight generalization of Proposition~\ref{prop:vc}.
\begin{prop}\label{prop:vc1}
Let $V_1,\dots, V_k$ be an ordered decomposition of $V$.  Let $\big$ denote the set of indices $i\in[k]$ such that $V_i$ is big in $V$.
 Assume $E$ is chosen uniformly at random (with repetitions) as a sample of $m$ elements from $\bigcup_{i\in\big}{V_i\choose 2}$, where $m > n$.  For each $i=1,\dots, k$, let $E_i = E \cap {V_i\choose 2}$.
Define $C_E(\pi, \{V_1, \dots, V_k\}, W)$
to be
\begin{equation}\label{eq:defCEdecomp}
 C_{E}(\pi, \{V_1, \dots, V_k\}, W) =  \left(\sum_{i\in\big} {n_i\choose 2}\right) |E|^{-1} \sum_{i\in\big} {n_i\choose 2}^{-1} |E_i| C_{E_i}(\pi_{|V_i}, V_i, W_{|V_i})\ .\end{equation} 
(The normalization is defined so that the expression is an unbiased estimator of $\sum_{i\in\big} C(\pi_{|V_i}, V_i, W_{|V_i})$.  If $|E_i|=0$ for some $i$,  formally define ${n_i\choose 2}^{-1} |E_i|C_{E_i}(\pi_{|V_i}, V_i, W_{|V_i})=0$.)
Then with  probability at least $1-e^{-n}$ over the sample, all permutations $\pi\in\Pi(V)$ satisfy:
\begin{equation*}
\left | C_E(\pi,\{V_1,\dots, V_k\},W) - \sum_{i\in\big} C(\pi|_{V_i},V_i,W|_{V_i})\right| =  \sum_{i\in\big} {n_i\choose 2}O\left ( \sqrt {\frac{n\log m + \log(1/\delta)}{m}}\right ) \ .
\end{equation*}
\end{prop}
\begin{proof}
Consider the set of binary functions $\prod_{i\in \big} \Pi(V_1)$ on the domain  $\bigcup_{i\in\big} V_i\times V_i$,
defined as follows:  If  $u,v\in V_j\times V_j$  for some $j\in\big$, then 
$$ \left((\pi_i)_{i\in\big}\right)(u,v) = \pi_j(u,v)\ .$$
It is clear that the VC dimension of this function set is at most the
sum of the VC dimensions of $\{\Pi(V_i)\}_{i\in\big}$, hence by Lemma~\ref{lem:linvcbound} at most $n$.  The result follows.
\end{proof}
\subsection{Using an $\eps$-Good Partition}
\noindent
The following lemma explains why an $\eps$-good partition is good for our purpose.
\begin{lem}\label{lem:whygoodisgood}
Fix $\eps>0$ and assume we have an $\eps$-good partition (Definition~\ref{def:goodpart}) $V_1,\dots, V_k$ of $V$. 
Let $\big$ denote the set of $i\in[k]$ such that $V_i$ is big in $V$, and let $\bar \big = [k]\setminus \big$. Let $n_i=|V_i|$ for $i=1,\dots, n$,  and let $E$ denote  a random sample of $O( \eps^{-6} n \log n)$ elements from $\bigcup_{i\in\big} {V_i\choose 2}$, each element chosen uniformly at random with repetitions.  Let $E_i$ denote $E\cap {V_i\choose 2}$.  
Let $C_E(\pi, \{V_1, \dots, V_k\}, W)$ be defined as in (\ref{eq:defCEdecomp}).
For any $\pi\in\Pi(V_1,\dots, V_k)$ define:
$$ 
\tilde C(\pi) :=  C_{E}(\pi, \{V_1, \dots, V_k\}, W) +  \sum_{i\in\bar\big} C(\pi_{|V_i}, V_i, W_{|V_i}) + \sum_{1\leq i<j\leq k} \sum_{(u,v)\in V_i\times V_j} {\one}_{v \prec_\pi u} \ .$$
Then the following event occurs with probability at least $1-e^{-n}$: For all $\sigma\in \Pi(V_1,\dots, V_k)$, 
\begin{equation}\label{eq:mainmleq}
\left| \tilde C(\sigma) - C(\sigma, V, W)\right| \leq \eps \min_{\pi\in\Pi(V)} C(\pi, V, W)\ .
\end{equation}
Also, if $\sigma^*$ is any minimizer of $\tilde C(\cdot)$ over $\Pi(V_1,\dots, V_k)$, then
\begin{equation}\label{eq:mainmleq1} C(\sigma^*, V, W)\leq (1+2\eps)\min_{\pi\in\Pi(V)} C(\pi, V, W)\ .\end{equation}
\end{lem}

Before we prove the lemma, let us discuss its consequences:  Given an $\eps$-good decomposition $V_1,\dots, V_k$ of $V$,
the theorem implies that if we could optimize $\tilde C(\sigma)$ over $\sigma\in\Pi(V_1,\dots, V_k)$, 
we would obtain a permutation $\pi$ with a \emph{relative regret} of $2\eps$ with respect to the optimizer of $C(\cdot, V, W)$ over $\Pi(V)$.  Optimizing $\sum_{i\in\hat\big} C(\pi_{|V_i}, V_i, W_{|V_i})$ is easy:  Each $V_i$ is
of size at most $\log n/\log\log n$, hence exhaustively searching its corresponding permutation space can be done in polynomial time.  In order to compute the cost of each permutation inside the small sets $V_i$, we would need to query $W_{|V_i}$ in its entirety.  This incurs a query cost of at most $\sum_{i\in\bar\big} {n_i\choose 2} = O(n\log n/\log\log n)$, which is dominated by the cost of
obtaining the $\eps$-good partition in the first place (see next sect section).
Optimizing $C_E(\pi, \{V_1,\dots, V_k\}, W)$ given $E$ is a tougher nut to crack, and is known as the
minimum feedback arc-set (MFAS) problem and considered much harder than
to harder than MFAST \cite{Karp, DS}.  For now we focus on query and not computational complexity, and notice that the size $|E|=O(\eps^{-4}n\log n)$ of the sample
set is all we need.  In Section~\ref{sec:svm} we show a counterpart of Lemma~\ref{lem:whygoodisgood}
which provides similar guarantees for practitioners who choose to relax it using SVM, for which fast solvers exist.

If we assume, in addition, that the decomposition could be computed using $O(n \polylog(n,\eps^{-1}))$ labels (as we
indeed show in the next section), then we would clearly beat the aforementioned  VC bound whenever the optimal solution $\min_{\pi\in \Pi(V)} C(\pi, V, W)$ is at most $O(n^{2-\nu})$, for any $\nu>0$.


\begin{proof}
For any permutation $\sigma\in\Pi(V_1,\dots, V_k)$, it is clear
that $$\tilde C(\sigma) - C(\sigma, V, W) = C_E(\sigma, \{V_1,\dots, V_k\}, W) - \sum_{i\in\big} C(\sigma_{|V_i}, V_i, W_{|V_i})\ .$$
By Proposition~\ref{prop:vc1}, with probability at least $1-e^{-n}$ the absolute value of the RHS is bounded by $\eps^3 \sum_{i\in\big}{n_i\choose 2}$, which is at most $\eps \min_{\pi
\in\Pi(V)} C(\pi, V, W)$ by (\ref{eq:localchaos}).  This establishes (\ref{eq:mainmleq}).
Inequality  (\ref{eq:mainmleq1})  is obtained from (\ref{eq:mainlmeq}) together with (\ref{eq:globalorder})
and the triangle inequality.

%
%
%
\end{proof}

\negspaceB
\section{A Query Efficient Algorithm for $\eps$-Good Decomposing}\label{sec:active}
\negspaceA
\def\optoracle{\operatorname{OptOracle}}
\noindent
The section is dedicated to proving the following:
\begin{thm}\label{thm:main}
Given a set $V$ of size $n$, a preference oracle $W$ and an error tolerance parameter $0<\eps<1$,
there exists a polynomial time algorithm which returns, with constant probabiliy, an $\eps$-good partition of $V$, querying at most
$O(\eps^{-6}n\log^5 n)$ locations in $W$ on expectation.  The running time of the algorithm (counting computations) is $O(n \polylog(n,\eps^{-1}))$.
\end{thm}
\noindent
\noindent
Before describing our algorithm, we need  some definitions.



\def\ra{\rightarrow}
\begin{defn}
Let $\pi$ denote a permutation over $V$.  Let $v\in V$ and $i\in [n]$.  We define  $\pi_{v\ra i}$ to be the permutation obtained
by moving the rank of $v$ to $i$ in $\pi$, and leaving the rest of the elements in the same order.  For example, if $V=\{x,y,z\}$
and $(\pi(1), \pi(2), \pi(3)) = (x,y,z)$, then $(\pi_{x\rightarrow 3}(1), \pi_{x\rightarrow 3}(2), \pi_{x\rightarrow 3}(3)) = (y,z,x)$.
\end{defn}
\def\test{\operatorname{TestMove}}
\begin{defn}
Fix a permutation $\pi$ over $V$, an element $v\in V$  and an integer $i\in [n]$.  We define the number $\test(\pi, V, W, v, i)$ as the  decrease in
 the cost $C(\cdot, V, W)$ achieved
  by moving from $\pi$ to $\pi_{v\rightarrow i}$. 
  More precisely, 
  $ \test(\pi, V, W, v, i) = C(\pi, V, W) - C(\pi_{v\rightarrow i}, V, W)\ .$
  Equivalently, if $i\geq \rank_\pi(v)$ then  $$\test(\pi, V, W, v, i) = \sum_{u: \rank_\pi(u) \in [\rank_\pi(v)+1 ,i]} (W_{uv}-W_{vu})\ .$$
    A similar expression can be written for $i < \rank_\pi(v)$.

    Now assume that we have a multi-set $E\subseteq {V\choose 2}$. 
    We define $\test_E(\pi, V, W, v, i)$, for $i \geq \rank_\pi(v)$, as
\negspaceeq
    \begin{eqnarray*}
    \test_E(\pi, V, W, v, i) &=& \frac {|i-\rank_\pi(v)|}{|\tilde E|}       \sum_{{u:(u,v)\in \tilde E}} (W(u,v)-W(v,u))\ ,
    \end{eqnarray*}
where the multiset $\tilde E$ is defined as  $\{(u,v)\in E: \rank_\pi(u) \in [\rank_\pi(v)+1, i]\}$.  

    Similarly, for $i < \rank_\pi(v)$ we define
\negspaceeq
    \begin{eqnarray}\label{eq:deftestE}
    \test_E(\pi, V, W, v, i) &=& \frac {|i-\rank_\pi(v)|}{|\tilde E|}
     \sum_{u:(u,v)\in \tilde E}(W(v,u)-W(u,v))\ ,
    \end{eqnarray}
where the multiset  $\tilde E$ is now defined as  $\{(u,v)\in E: \rank_\pi(u) \in [i,\rank_\pi(v)-1]\}$.  

\end{defn}
\noindent

\def\E{{\bf E}}
\begin{lem}
Fix a permutation $\pi$ over $V$, an element $v\in V$, an integer $i\in [n]$ and another integer $N$.  Let $E\subseteq {V\choose 2}$ be a random (multi)-set of size $N$
with elements  $(v,u_1),\dots, (v,u_N)$, drawn so that for each $j\in [N]$ the element $u_j$ is chosen uniformly at random from among the elements lying between
$v$ (exclusive) and position $i$  (inclusive) in $\pi$.
Then $\E[\test_E(\pi, V, W, v, i)] = \test(\pi, V, W, v, i)$.  Additionally, for any $\delta>0$, except with probability of failure $\delta$,
\negspaceeq
$$ |\test_E(\pi, V, W, v, i) - \test(\pi, V, W, v, i)| = O\left (  |i-\rank_\pi(v)|\sqrt{\frac {\log \delta^{-1}}{N}} \right )\ .$$
\end{lem}
\noindent
The lemma is easily proven using e.g. Hoeffding tail bounds, using the fact that $|W(u,v)| \leq 1$ for all $u,v$.

\subsection{The Decomposition Algorithm}\label{sec:improved}

\def\qsout{{\pi}}  
\def\inperm{\pi}   
\def\approxC{C}  
\def\tree{{\mathcal T}}
\def\improveout{{\pi_1}}  
\def\costout{{C_1}}  
\def\piL{\pi_L}  
\def\piR{\pi_R}  

Our decomposition algorithm  $\samplerankroot$ is detailed in Algorithm~\ref{fig:samplerankroot}, with subroutines
in Algorithms~\ref{fig:samplerank} and~\ref{fig:localimprove}.
It can be viewed as a query efficient improvement of the main algorithm in \cite{DBLP:conf/stoc/Kenyon-MathieuS07}.
Another difference is that we are not interested in an approximation algorithm for MFAST:  Whenever we reach a
small block (line~\ref{line:small}) or a big block
with a probably approximately sufficiently high cost  (line~\ref{line:early}) 
in our recursion  of Algorithm~\ref{fig:samplerank}), 
we simply output
it as a block in our partition.  Denote the resulting outputted partition by $V_1,\dots, V_k$.
Denote by $\hat \pi$ the minimizer
of $C(\cdot, V, W)$ over $\Pi(V_1,\dots, V_k)$.
Most of the analysis is dedicated to showing that $C(\hat\pi, V, W) \leq (1+\eps)\min_{\pi\in\Pi(V)} C(\pi, V, W)$, thus establishing (\ref{eq:globalorder}).

In order to achieve an efficient query complexity compared to  \cite{DBLP:conf/stoc/Kenyon-MathieuS07},
we use  procedure $\localimprove$ (Algorithm~\ref{fig:localimprove}) to replace a greedy local improvement step in  \cite{DBLP:conf/stoc/Kenyon-MathieuS07} which is \emph{not} query efficient.
Aside from the aforementioned differences, we also raise here the reader's awareness to the query efficiency of $\operatorname{QuickSort}$, which was established by Ailon et al. in \cite{DBLP:jmlr/AilonM08}
(note: an erroneous proof appears in \cite{DBLP:conf/colt/AilonM08}).

$\samplerankroot$ (Algorithm~\ref{fig:samplerankroot})  takes the following arguments: The set $V$ we want to rank, the preference matrix $W$ and an accuracy argument $\eps$.
 It is implicitly understood that the argument $W$ passed to $\samplerankroot$ is given as a query oracle, incurring a unit cost upon each access to a matrix element by the procedure and any nested calls.

The first step in $\samplerankroot$ is to obtain an expected constant factor approximation  $\qsout$ to MFAST on $V,W$, incurring an expected low query cost.
More precisely, this step returns a random permutation $\qsout$ with an expected cost of $O(1)$ times that of the optimal solution
to MFAST on $V,W$.  The query complexity of this step is $O(n\log n)$ \emph{on expectation} \cite{DBLP:jmlr/AilonM08}.  Before continuing, we make the following
assumption, which holds with constant probability using Markov probability bounds.
\begin{assmp}\label{assmp:qs}
The cost $C(\qsout, V, W)$ of the initial permutation $\qsout$ computed line~\ref{line:qs} of $\samplerankroot$ is at most $O(1)$ times that of the optimal
solution $\pi^*$ to MFAST on $(V,W)$, and the query cost incurred in the computation is $O(n\log n)$.
\end{assmp}

Following $\operatorname{QuickSort}$, a recursive procedure  $\samplerank$ is called.  It implements a 
divide-and-conquer algorithm.
Before branching, it executes the following steps.
Lines~\ref{line:C0} to~\ref{line:endearly} are responsible for identifying local chaos, with sufficiently high probability.
The following line~\ref{line:afterif} calls a procedure $\localimprove$ (Algorithm~\ref{fig:localimprove}) which is responsible for performing
query-efficient approximate greedy steps.  We devote the next Sections~\ref{sec:localimprove}-\ref{sec:localimproveqq} to describing this procedure.
The establishment of the $\eps$-goodness of $\samplerankroot$'s output (establishing (\ref{eq:globalorder})) is deferred to Section~\ref{sec:analysis}.

\negspaceC
\subsection{Approximate local improvement steps}\label{sec:localimprove}
\negspaceD
 The procedure $\localimprove$  takes as input a set $V$ of size $N$, the preference oracle $W$, a permutation $\pi$ on $V$, two
numbers $C_0$, $\eps$ and an integer $n$. The number $n$ is the size of the input in the root call to $\samplerank$, passed down in the recursion,
and used for the purpose of controlling the success probability of each call to the procedure (there are a total of $O(n\log n)$ calls,
and a union bound will be used to bound a failure probability, hence each call may fail with probability inversely polynomial in $n$).
The goal of the procedure is to repeatedly identify, with high probability, single vertex moves that considerably decrease the cost.
 Note that in Mathieu et. al's PTAS \cite{DBLP:conf/stoc/Kenyon-MathieuS07},
a crucial step in their algorithms entails identifying single vertex moves that decrease 
the cost by a magnitude which, given our sought query complexity, would not be detectable.
Hence, our algorithm requires altering this crucial part in their algorithm.

The procedure starts by creating a \emph{sample ensemble} $\S = \{E_{v,i}: v\in V, i \in [B, L]\}$,
where $B = \log\lfloor \Theta(\eps N/\log n)\rfloor$ and $L=\lceil \log N\rceil$.
The size of  each $E_{v,i}\in \S$ is $\Theta(\eps^{-2}\log^2n)$, and each element $(v,x)\in E_{v,i}$ was added (with possible multiplicity) by uniformly at random
selecting, with repetitions, an element $x\in V$
positioned at distance at most $2^i$ from the position of $v$ in $\pi$.  Let $\D_\pi$ denote the distribution space from which $\S$ was drawn,
 and let $\Pr_{X\sim \D_\pi}[X=\S]$ denote the probability of obtaining a given sample ensemble $\S$.

We want $\S$ to enable us to approximate the improvement in cost obtained by moving a single element $u$ to position $j$.  
\begin{defn}\label{defn:successful}
Fix $u\in V$ and $j\in [n]$, and assume $\log|j-\rank_\pi(u)| \geq B$. 
Let $\ell = \lceil \log |j - \rank_\pi(u)|\rceil$.  We say that $\S$ is successful at $u,j$ if
 $\left |\{x: (u,x) \in E_{u,\ell}\} \cap \{x: \rank_\pi(x) \in [\rank_\pi(u), j]\}\right | = \Omega(\eps^{-2}\log^2 n)\ .$
\end{defn}
In words, success of $\S$ at $u,j$ means that sufficiently many samples $x\in V$ such that $\rank_\pi(x)$ is between $\rank_\pi(u)$ and $j$
are represented in $E_{u, \ell}$.
Conditioned on $\S$ being successful at $u,j$, note that  the denominator of $\test_E$ (defined in (\ref{eq:deftestE})) does not vanish,
and we can thereby define:
\begin{defn}
$\S$ is a \emph{good approximation} at $u,j$ if
\negspaceA
$$\left |\test_{E_{u,\ell}}(\pi, V, W, u, j) - \test(\pi, V, W, u, j)\right | \leq \frac 1 2 \eps|j-\rank_\pi(u)| / \log n\ ,$$
where $\ell$ is as in Definition~\ref{defn:successful}.
\end{defn}
In words, $\S$ being a good approximation at $u,j$ allows us to approximate a quantity of interest $\test(\pi, V, W, u, j)$, and to detect
whether it is sufficiently large, and more precisely, at least $\Omega(\eps|j-\rank_\pi(u)|/\log n)$.
\begin{defn}\label{defn:Sgoodapprox}
We say that $\S$ is a good approximation if it is succesful and a good approximation at all $u\in V$, $j\in [n]$ satisfying $\lceil \log  |j-\rank_\pi(u)|\rceil \in [B,L]$.
\end{defn}
\noindent Using \neghspace Chernoff bounds to ensure \neghspace  that $\S$ is  successful  \neghspace $\forall u,j$ \neghspace as in \neghspace Definition \ref{defn:Sgoodapprox}, \neghspace then  \neghspace  using Hoeffding to ensure
that \neghspace $\S$ is a good approximation \neghspace at all such $u,j$ and finally union bounding \neghspace we get
\begin{lem}\label{lem:Sgood}
Except with probability $1-O(n^{-4})$, $\S$ is a good approximation.
\end{lem}

\begin{algorithm}[h!]
\caption{$\samplerankroot(V, W, \eps)$}\label{fig:algroot}
\begin{algorithmic}[1]
 \STATE $n\leftarrow |V|$ \label{line:initializen}
 \STATE $\qsout \leftarrow$ Expected $O(1)$-approx solution to MFAST using $O(n \log n)$ $W$-queries  on expectation using QuickSort \cite{DBLP:journals/jacm/AilonCN08} \label{line:qs}
  \RETURN $\samplerank(V, W, \eps, n,\qsout)$\label{line:calltosamplerank}
\end{algorithmic}
\label{fig:samplerankroot}
\end{algorithm}

\begin{algorithm}[h!]
\caption{$\samplerank(V, W, \eps, n, \inperm)$}\label{fig:samplerank}
\begin{algorithmic}[1]
\STATE $N \leftarrow |V|$
\IF {$N \leq \log n / \log\log n$}\label{line:smallif} \RETURN trivial partition $\{V\}$ \label{line:small} \ENDIF
\STATE $E \leftarrow$ random subset of $O(\eps^{-4}\log n)$ elements from ${V\choose 2}$ (with repetitions) \label{line:C0} 
\STATE $\approxC \leftarrow C_E(\inperm, V, W)$  \;\;\;\;\;  ($C$ is an additive $O(\eps^2 N^2)$ approximation of $C$ w.p. $\geq 1-n^{-4}$)
    \IF {$\approxC = \Omega(\eps^2 N^2)$}\label{line:omega}
        \RETURN  trivial partition $\{V\}$  \label{line:early}  
     \ENDIF\label{line:endearly}
 \STATE $\improveout \leftarrow \localimprove(V, W, \inperm,  \eps, n)$ \label{line:afterif}
 \STATE $k \leftarrow $ random integer in the range $[N/3,2N/3]$ \label{line:k}
 \STATE $V_L \leftarrow \{v\in V: \rank_\pi(v) \leq k\}$, $\piL \leftarrow $ restriction of $\improveout$ to $V_L$ \label{line:recurseL}
 \STATE $V_R \leftarrow V\setminus V_L$,\ \ \ \ \ \ \ \ \ \ \ \ \ \ \ \ \ \  $\piR \leftarrow $ restriction of $\improveout$ to $V_R$\label{line:recurseR}
 \RETURN concatenation of decomposition $\samplerank(V_L, W, \eps, n, \piL)$ and decomposition $\samplerank(V_R, W, \eps, n, \piR)$ \label{line:return}
\end{algorithmic}
\end{algorithm}

\begin{algorithm}[h!]
\caption{$\localimprove(V, W, \pi,  \eps, n)$  ({\emph Note}: $\pi$  used as both input and output)}\label{fig:localimprove}
\begin{algorithmic}[1]
\STATE $N\leftarrow |V|$, $B\leftarrow \lceil \log(\Theta(\eps N/\log n)ד\rceil$, $L \leftarrow \lceil \log N\rceil$\label{line:consts}
 \IF {$N = O(\eps^{-3}\log^3 n)$}
\RETURN  \label{line:exitnothing}  \ENDIF
\FOR{$v\in V$} \label{line:dist1}
 \STATE $r \leftarrow \rank_\pi(v)$
 \FOR{$i = B \dots L$}\label{line:iloop}
   \STATE $E_{v,i}\leftarrow \emptyset$
   \FOR{$m= 1..\Theta(\eps^{-2}\log^2 n)$}
       \STATE $j \leftarrow $ integer uniformly at random chosen from
            $[\max\{1,r-2^i\}, \min\{n,r+2^i\}]$\label{line:clip}
       \STATE $E_{v,i} \leftarrow E_{v,i} \cup \{(v, \pi(j))\}$
   \ENDFOR
 \ENDFOR
\ENDFOR\label{line:dist2}
\WHILE{$\exists u\in V$ and  $j\in [n]$ s.t. \emph{(setting $\ell :=\lceil \log |j-\rank_\pi(u)|\rceil$):}\negspaceA
\begin{eqnarray*}\ell \in [B,L] \mbox{ and }\test_{E_{u,\ell }}(\pi, V, W, u, j) >  \eps|j-\rank_\pi(u)|/\log n\end{eqnarray*}}\label{line:while}
    \FOR{$v\in V$ and $i\in [B,L]$}
       \STATE  refresh sample $E_{v, i}$ with respect to the move $u\rightarrow j$ (see Section~\ref{sec:mutate}) \label{line:mutate} 
    \ENDFOR
    \STATE $\pi \leftarrow \pi_{u\rightarrow j}$
\ENDWHILE\label{line:endwhile}
\end{algorithmic}
\end{algorithm}

\negspaceB
\subsection{Mutating the Pair Sample To Reflect a Single Element Move}\label{sec:mutate}
\negspaceA
Line~\ref{line:mutate} in $\localimprove$ requires elaboration.
  In lines~\ref{line:while}-\ref{line:endwhile}, we check whether there exists an element $u$
and position $j$, such that moving $u$ to $j$ (giving rise to $\pi_{u\rightarrow j}$) would considerably improve the MFAST cost of the
procedure input, based on a high probability approximate
calculation.
 The approximation is done using the sample ensemble $\S$.
If such an element $u$ exists, we execute  the exchange $\pi \leftarrow \pi_{u\rightarrow j}$.  With respect to the new value of the permutation $\pi$,
 the sample ensemble $\S$ becomes \emph{stale}.
By this we mean, that if $\S$ was a good approximation with respect to $\pi$, then it is no longer necessarily a good approximation with respecto to $\pi_{u\rightarrow j}$.
  We must refresh it.
Before the next iteration of the while loop, we perform in line~\ref{line:mutate}
a  transformation $\varphi_{u\rightarrow j}$ to $\S$, so that the resulting sample ensemble $\varphi_{u\rightarrow j}(\S)$ is distributed according to $\D_{\pi_{u\rightarrow j}}$.
More precisely, we will define a transformation $\varphi$ such that
\negspaceA
\def\totvar{d_{\operatorname{TV}}} 
\begin{equation}\label{eq:chain}
\varphi_{u\rightarrow j}(\D_\pi) =  D_{\pi_{u\rightarrow j}}\ ,
\end{equation}
where the left hand side denotes the distribution obtained by drawing from $\D_\pi$ and applying $\varphi_{u\rightarrow j}$ to the result.
The transformation $\varphi_{u\rightarrow j}$ is performed as follows.  Denoting $\varphi_{u\rightarrow j}(\S) = \S' = \{E'_{v,i}: v\in V, i\in [B,L]\}$, 
we need to define each $E'_{v,i}$.
\begin{defn}\label{defn:interesting}
\def\intsetA{T_1}
\def\intsetB{T_2}
We say that $E_{v,i}$ is \emph{interesting} in the context of $\pi$ and $\pi_{u\rightarrow j}$  if the two sets $\intsetA, \intsetB$ defined as
\negspaceeq
\begin{eqnarray}
 \intsetA &=& \{x \in V: |\rank_\pi(x) - \rank_\pi(v)| \leq 2^i \}\\
 \intsetB &=& \{x \in V: |\rank_{\pi_{u\rightarrow j}}(x) - \rank_{\pi_{u\rightarrow j}}(v)| \leq 2^i\}
\end{eqnarray}
differ.
\end{defn}
\noindent
We set $E'_{v,i} = E_{v,i}$ for all  $v,i$ for which  $E_{v,i}$ is \emph{not} interesting.

\begin{obs}\label{obs:interesting}
There are at most
$O(|\rank_\pi(u)-j|\log n)$ interesting choices of $v,i$.
Additionally, if $v\neq u$, then for $T_1, T_2$ as in Definition~\ref{defn:interesting}, $|T_1\Delta T_2| = O(1)$, where $\Delta$ denotes symmetric difference.
\end{obs}
 Fix one interesting choice $v,i$. Let $T_1, T_2$ be as in Defintion~\ref{defn:interesting}.  By the last observation,
 each of $T_1$ and $T_2$ contains $O(1)$ elements that are not
contained in the other.  Assume $|T_1|=|T_2|$, let $X_1 = T_1\setminus T_2$, and $X_2 = T_2\setminus T_1$.  Fix any injection $\alpha: X_1\rightarrow X_2$, and
extend $\alpha : T_1\rightarrow T_2$ so that $\alpha(x) = x$ for all $x\in T_1\cap T_2$.  Finally, define
\negspaceeq
\begin{equation}\label{eq:mutationcases}
 E'_{v,i} = \{ (v, \alpha(x)) : (v,x) \in E_{v,i}\}\ .
\end{equation}
\def\pseudoright{{\tilde v}^R}
\def\pseudoleft{{\tilde v}^L}
(The case  $|T_1|\neq |T_2|$ may occur due to the clipping of the ranges $[\rank_\pi(v)-2^i, \rank_\pi(v)+2^i]$ and $[\rank_{\pi_{u\rightarrow j}}(v)-2^i, \rank_{\pi_{u\rightarrow j}}ד(v)+2^i]$
to a smaller range.  This is a simple technicality which may be taken care of by formally extending the set $V$ by $N$ additional  elements 
$\pseudoleft_1,\dots, \pseudoleft_N$, extending the definition of $\rho_\pi$ for all permutation $\pi$ on $V$ so that $\rho_\pi(\pseudoleft_a)=-a+1$ for all $a$
and similarly $N=|V|$ additional elements $\pseudoright_1,\dots, \pseudoright_N$ such that $\rho_\pi(\pseudoright_a)=N+a$.
Formally extend $W$ so that $W(v,\pseudoleft_a)=W(\pseudoleft_a,v)=W(v,\pseudoright_a)=W(\pseudoright_a,v)=0$ for all $v\in V$ and $a$.
This eliminates  the need for clipping ranges in line~\ref{line:clip} in $\localimprove$.)

Finally, for $v=u$ we create  $E'_{v,i}$ from scratch by repeating the loop in line~\ref{line:iloop}  for that $v$.
\noindent
\def\dist{\operatorname{dist}}
It is easy to see that (\ref{eq:chain}) holds.  We need, however, something stronger that $(\ref{eq:chain})$.  Since our analysis
assumes that $\S\sim \D_\pi$ is successful, we must be able to measure the distance (in total variation) between the random variable
$(\D_\pi|\mbox{ success})$ defined by the process of drawing from $\D_\pi$ and conditioning on the result's success, and $\D_{\pi_{u\rightarrow j}}$.
By Lemma~\ref{lem:Sgood}, the total variation distance between $(\D_\pi|\mbox{ success})$ and $\D_{\pi_{u\rightarrow j}}$ is $O(n^{-4})$.
Using a simple chain rule argument, we conclude the following:
\begin{lem}\label{lem:chain}
Fix $\pi^0$ on $V$ of size $N$, and fix $u_1,\dots, u_k\in V$ and $j_1,\dots, j_k\in [n]$.
Consider the following process.  We draw $\S^0$ from $\D_{\pi^0}$, and define 
\begin{eqnarray*}
\S^1=\varphi_{u_1\rightarrow j_1}(\S^0), S^2=\varphi_{u_2\rightarrow j_2}(\S^1),&\cdots&, S^k=\varphi_{u_k\rightarrow j_k}(\S^{k-1}) \\
\pi^1 = \pi^0_{u_1\rightarrow j_1}, \pi^2 = \pi^1_{u_2\rightarrow j_2},&\cdots&, \pi^k = \pi^{k-1}_{u_k \rightarrow j_k} \ .
\end{eqnarray*}
Consider the random variable $S^k$ conditioned on $S^0, S^1, \dots, S^{k-1}$ being successful for $\pi_0,\dots, \pi^{k-1}$, respectively.
Then the total variation distance between the distribution of $S^k$ and the distribution $\D_{\pi^k}$ is at most $O(kn^{-4})$.
\end{lem}
\negspaceC
\subsection{Bounding the query complexity of computing $\varphi_{u\rightarrow j}(\S)$}\label{sec:localimproveqq}
\negspaceD
We now need a notion of distance between $\S$ and $\S'$,  measuring how many extra pairs were introduced ino the new
sample family.  These pairs may incur the cost of querying $W$.  We denote this measure as $\dist(\S, \S')$, and define it as
  $\dist(\S, \S') := \left | \bigcup_{v,i} E_{v,i}\Delta E'_{v,i} \right |\ .$
\begin{lem}\label{lem:mutationcomplexity}
Assume $\S \sim \D_{\pi}$ for some permutation $\pi$, and $\S' = \varphi_{u\rightarrow j}$.  Then
$ \E[\dist(\S,\S')] = O(\eps^{-3}\log^3 n)$.  
\end{lem}
\begin{proof}
Denote $\S = \{E_{v,i}\}$ and $\S' = \{E'_{v,i}\}$.
Fix some $v\neq u$.  By construction, the  sets $E_{v,i}$ for which $E_{v,i} \neq E'_{v,i}$ must be interesting, and there are at  most $O(|\rank_\pi(u)-j|\log n)$ such, using Observation~\ref{obs:interesting}.
Fix such a choice of  $v,i$.  By (\ref{eq:mutationcases}), $E_{v,i}$ will indeed differ from $E'_{v,i}$ only if it contains an element $(v, x)$ for some $x\in T_1\setminus T_2$.
But the probability of that is at most
 $$1 - (1-O(2^{-i}))^{\Theta(\eps^{-2}\log^2 n)} \leq 1-e^{-\Theta(\eps^{-2}2^{-i}\log^2 n)} = O(\eps^{-2}2^{-i}\log^2 n) \ $$
(We used the fact that $i\geq B$, where $B$ is as defined in 
line~\ref{line:consts} of $\localimprove$, and $N =\Omega(\eps^{-3}\log^3 n)$ as guaranteed in line~\ref{line:exitnothing} of $\localimprove$.)
 Therefore, the expected size of $E'_{v,i}\Delta E_{v,i}$ (counted with multiplicities)  is $O(\eps^{-2}2^{-i}\log^2 n)$.

Now consider all the interesting sets $E_{v_1. i_1},\dots, E_{v_P, i_P}$.  For each possible value $i$ it is easy to see that there are at most $2|\rank_\pi(u)-j|$ $p$'s for which $i_p = i$.
Therefore,
$ \E\left [\sum_{p=1}^P |E'_{v_p,i_p}\Delta E_{v_p,i_p}|\right ] = O\left(\eps^{-2}|\rank_\pi(u) - j|\log^2 n\sum_{i=B}^L2^{-i}\right),$
where $B,L$ are defined in line~\ref{line:consts} in $\localimprove$.  Summing over $i\in [B,L]$, we get at most $O(\eps^{-3}|\rank_\pi(u) - j|\log^3 n / N)$.
For $v=u$, the set $\{E_{v,i}\}$ is drawn from scratch, clearly contributing $O(\eps^{-2}\log^3 n)$ to $\dist(\S, \S')$.
   The claim follows.
\end{proof}

\negspaceB
\subsection{Analysis of $\samplerank$}\label{sec:analysis}
\negspaceA
Throughout the execution of the algorithm, various \emph{high probability} events must occur in order for the algorithm guarantees to hold.
Let $\S_1, \S_2,\dots$ denote the sample families that are given rise to through the executions of $\localimprove$, either between lines~\ref{line:dist1} and~\ref{line:dist2}, or
as a mutation done between lines~\ref{line:while} and~\ref{line:endwhile}.  We will need the first $\Theta(n^4)$   to be good approximations, based on Definition~\ref{defn:Sgoodapprox}.
Denote this favorable event $\ES$.
By Lemma~\ref{lem:chain}, and using a union bound, with constant probability (say, $0.99$) this happens.
 We also need the cost approximation $C$  obtained in line~\ref{line:C0} to be successful.  Denote this favorable event $\Ecostapprox$.
By Hoeffding tail bounds, this happens with probability $1-O(n^{-4})$ for each
execution of the line.  This line is  obviously executed at most $O(n\log n)$ times, and hence we can lower bound the probability of success of all
executions by  $0.99$.

From now throughout, we make the following assumption, which is true by the above with probability at least $0.97$.
\begin{assmp}\label{assmp:events}
Events $\ES$ and $\Ecostapprox$ hold true.
\end{assmp}
Note that by conditioning the remainder of our analysis on this assumption may bias some expectation upper bounds derived earlier and in what follows.  This
bias can multiply the estimates by at most $1/0.97$, which can be absorbed in the $O$-notation of these bounds.

Let $\pi^*$ denote the optimal permutation for the root call to $\samplerank$ with $V,W,\eps$.
The permutation $\qsout$ is, by Assumption~\ref{assmp:qs},
 a constant factor approximation for MFAST on $V,W$.
Using the triangle inequality, we conclude that $\dtau(\qsout, \pi^*) \leq C(\qsout, V, W) + C(\pi^*, V, W)$  Hence,
 $E[\dtau(\qsout, \pi^*)] = O(C(\pi^*, V, W))\ .$
From this we conclude, using (\ref{eq:diaconis}), that
$$E[\dfoot(\qsout, \pi^*)] = O(C(\pi^*, V, W))\ .$$

Now consider the recursion tree $\tree$ of $\samplerank$.
Denote $\internal$ the
set of internal nodes, and by $\leaf$ the set of leaves (i.e. executions
exiting from line~\ref{line:early}).
For a call $\samplerank$ corresponding to a node $X$ in the recursion tree, denote the input arguments by $(V_X, W, \eps, n, \pi_X)$.  
Let $\leftc{X},\rightc{X}$ denote the  left and right children of $X$
respectively.  Let $k_X$ denote the integer $k$ in~\ref{line:k}
in the context of $X\in \internal$.  Hence, by our definitions, $V_{\leftc{X}}, V_{\rightc{X}},\pi_{\leftc{X}}$ and  $\pi_{\rightc{X}}$ are
precisely $V_L, V_R, \pi_L, \pi_R$ from lines~\ref{line:recurseL}-\ref{line:recurseR} in the context
of node $X$.

 Take, as in line~\ref{line:consts}, $N_X=|V_X|$.  Let $\pi^*_X$ denote the
optimal MFAST solution for instance $(V_X, W_{|V_X})$.
By $\ES$ we conclude that the first $\Theta(n^4)$ times in which we iterate through the while loop in $\localimprove$  (counted over all calls to $\localimprove$), the cost of ${\pi_{X}}_{u\rightarrow j}$
is an actual improvement compared to $\pi_X$ (for the current value of $\pi_X, u$ and $j$ in iteration), and the improvement in cost
is of magnitude at least
$\Omega(\eps|\rank_{\pi_X}(u)- j|/\log n),$ which is $\Omega(\eps^2 N_X/\log^2 n)$ due to the use of $B$ defined in line~\ref{line:consts}.  But this means that the
number of iterations
of the while loop in line~\ref{line:while} of $\localimprove$
is $O(\eps^{-2}C(\pi_X, V_X, W_{|V_X})\log^2 n/N_X)$. 
 Indeed, otherwise the true cost of the running solution would go below $0$.  Since $C(\pi_X, V_X, W_{|V_X})$ is at most ${N_X\choose 2}$, the number of iterations is hence at most $O(\eps^{-2}N_X\log^2 n)$.
By Lemma~\ref{lem:mutationcomplexity} the
expected  query complexity incurred by the call to $\localimprove$ is therefore
$ O(\eps^{-5}N_X \log^5 n )$.  Summing over the recursion tree, the total query complexity incurred by calls to $\localimprove$ is, on expectation,
at most $O(\eps^{-5} n\log^6 n)$.


\def\Vshort{V^{\operatorname{short}}}
\def\Vlong{V^{\operatorname{long}}}
\def\pio{{\pi_1}}
Now consider the moment at which the while loop of $\localimprove$ terminates.    Let $\improveout_X$ denote the permutation obtained at that point, returned to $\samplerank$
in line~\ref{line:afterif}.
We classify the elements $v\in V_X$ to two families:  $\Vshort_X$ denotes all $u\in V_X$ s.t.
 $|\rank_{\improveout_X}(u) - \rank_{\pi^*_X}(u)| =O(\eps N_X / \log n)$,
and $\Vlong_X$ denotes $V_X\setminus \Vshort_X$.
We know, by assumption, that the last sample ensemble $\S$ used in $\localimprove$ was a good approximation, hence for all $u\in \Vlong_X$,
\negspaceeq
\begin{equation}\label{eq:testbound}
\test(\improveout_X, V_X, W_{|V_X}, u, \rank_{\pi^*_X}(u)) = O(\eps|\rank_{\improveout_X}(u)-\rank_{\pi^*_X}(u)|/\log n). 
\end{equation}
\begin{defn}[Kenyon and Schudy \cite{DBLP:conf/stoc/Kenyon-MathieuS07}]
For $u\in V_X$, we say that $u$ crosses $k_X$ if the interval  $[\rank_{\improveout_X}(u), \rank_{\pi^*_X}(u)]$ contains the integer $k_X$.  
\end{defn}
\def\Vcross{V^{\operatorname{cross}}}
Let $\Vcross_X$ denote the (random) set of elements $u\in V_X$ that cross $k_X$ as chosen in line~\ref{line:k}.  We define a key quantity $T_X$ as in \cite{DBLP:conf/stoc/Kenyon-MathieuS07} as follows:
\begin{equation}\label{eq:defT}
 T_X = \sum_{u\in \Vcross_X} \test(\pio_X, V_X, W_{|V_X}, u, \rank_{\pi^*_X}(u)) \ .
 \end{equation}

Following (\ref{eq:testbound}), the elements $u\in \Vlong_X$ can contribute at most $$O\left (\eps\sum_{u\in \Vlong_X}|\rank_{\pio_X}(u) - \rank_{\pi^*_X}(u)| /\log n\right )$$ to $T_X$.
 Hence   the total contribution from such elements is, by definition $O(\eps\dfoot(\pio_X, \pi^*_X)/\log n)$ which is, using (\ref{eq:diaconis}) at most $O(\eps \dtau(\pio_X, \pi^*_X)/\log n)$.
Using the triangle inequality and the definition of $\pi^*_X$, the last expression, in turn, is
at most $O(\eps C(\pio_X, V_X, W_{|V_X})/\log n)$.

We now bound the contribution of the elements $u\in \Vshort_X$ to $T_X$.  The probability of each such element to cross $k$ is $O(|\rank_{\pio_X}(u) - \rank_{\pi^*_X}(u)|/N_X)$.
Hence, the total expected contribution of these elements to $T_X$ is \begin{equation}\label{eq:short} O\left (\sum_{u\in \Vshort_X} |\rank_{\pio_X}(u) - \rank_{\pi^*_X}(u)|^2/N_X \right )\ .\end{equation}
Under the constraints $\sum_{u\in \Vshort_X} |\rank_{\pio_X}(u) - \rank_{\pi^*_X}(u)| \leq \dfoot(\pio_X, \pi^*_X)$ and $ |\rank_{\pio_X}(u) - \rank_{\pi^*_X}(u)| =O(\eps N_X/\log n)$, the maximal value
of (\ref{eq:short}) is $$O(\dfoot(\pio_X, \pi^*_X) \eps N_X/(N_X \log n)) = O(\dfoot(\pio_X, \pi^*_X) \eps /\log n)\ .$$  Again using (\ref{eq:diaconis}) and the triangle
inequality, the last expression is  $O(\eps C(\pio_X, V_X, W_{|V_X}) / \log n)$.

\noindent
Combining the accounting for $\Vlong$ and $\Vshort$, we conclude
\begin{equation}\label{eq:Tbound}
E_{k_X}[T_X] = O(\eps C(\pi^*_X, V_X, W_{|V_X})/\log n)\ ,
\end{equation}
where the expectation is over 
the choice of $k_X$ in line~\ref{line:k} of $\samplerank$.

We are now in a  position to use a key Lemma from Kenyon et al's work \cite{DBLP:conf/stoc/Kenyon-MathieuS07}.  First we need a definition:  Consider the optimal solution $\pi'_X$ respecting ${V_{\leftc{X}}}, {V_\rightc{X}}$ in lines~\ref{line:recurseL} and~\ref{line:recurseR}.  By this we mean that $\pi'_X$ must rank all of the elements in ${V_X}_L$
before (to the left of) ${V_R}_X$. For the sake of brevity, let $C^*_X$ be shorthand for $C(\pi^*_X, V_X, W_{|V_X})$ 
and $C'_X$ for $C(\pi^;_X, V_X, W_{|V_X})$.

\begin{lem}\label{lem:fromkenschud}[Kenyon and Schudy \cite{DBLP:conf/stoc/Kenyon-MathieuS07}]
With respect to the distribution of the number $k_X$ in line~\ref{line:k} of $\samplerank$,
\negspaceeq \begin{equation}\label{eq:key}
 E[C'_X] \leq O\left (\frac {\dfoot(\pio_X, \pi^*_X)^{3/2}}{N_X}\right ) + E[T_X] + C^*_X\ .
\end{equation}
\end{lem}

Using (\ref{eq:diaconis}), we can replace $\dfoot(\pio_X, \pi^*_X)$ with $\dtau(\pio_X, \pi^*_X)$ in (\ref{eq:key}).  Using the triangle inequality, we
can then, in turn, replace $\dtau(\pio_X, \pi^*_X)$ with $C(\pio_X, V_X, W_{|V_X})$.


\subsection{Summing  Over the Recursion Tree}

Let us study the implication of  (\ref{eq:key}) for our purpose.
Recall that $\{V_1,\dots, V_k\}$ is the decomposition returned
by $\samplerankroot$, where  each $V_i$ corresponds to a leaf
in the recursion tree. 
Also recall that  $\hat \pi$ denotes the minimizer of $C(\cdot, V, W)$ over all    permutations in $\Pi(V_1,\dots, V_k)$ respecting the decomposition.  
Given Assumption~\ref{assmp:events} it suffices, for our purposes, to  show  that $\hat \pi$ is a (relative) small approximation for MFAST on $V,W$.
Our analysis of this account is basically that of \cite{DBLP:conf/stoc/Kenyon-MathieuS07}, with slight changes stemming from bounds
we derive on $E[T_X]$.  We present the proof in full detail for the sake of completeness.
Let $\root$ denote the root node.

For $X\in \internal$, let $\beta_X$ denote the contribution
of the split $\leftc{X}, \rightc{X}$ to the LHS of (\ref{eq:globalorder}).  More precisely, 
$$ \beta_X = \sum_{u \in \leftc{X}, v\in \rightc{X}} \one_{W(v,u)=1}\ ,$$
so we get $\sum_{1\leq i<j\leq k}\sum_{(u,v)\in V_i\times V_j} \one_{W(v,u)=1} = \sum_{X\in \internal} \beta_X$.
\noindent

For any $X\in \internal$, note also that by our definitions $\beta_X = C'_X - C^*_{\leftc{X}} - C^*_{\rightc{X}}$.
Hence, using Lemma~\ref{lem:fromkenschud} and the
ensuing comment,
$$ E[\beta_X] \leq O\left (E\left [ \frac {C(\improveout_X, V_X, W_{|V_X})^{3/2}}{N_X}\right ]\right ) + E[T_X] + E[C^*_X] - E[C^*_{\leftc{X}}] - E[C^*_{\rightc{X}}] \ ,$$
where the expectations are over the enitre space of random decisions made by the algorithm execution.
Summing the last inequality over $X\in \internal$, we get
(minding the cancellations):
\begin{eqnarray}\label{eq:difficult}
E\left[\sum_{X\in \internal} \beta_X\right] &\leq& O\left (\sum_{X\in \internal}E\left [ \frac {C(\improveout_X, V_X, W_{|V_X})^{3/2}}{N_X}\right ]\right ) + E\left[\sum_{X\in \internal}T_X\right] + C^*_\root - \sum_{X\in \leaf} E[C^*_X]\ .
\end{eqnarray}

The expression $E[\sum_{X\in\internal} T_X]$ is bounded by
$O\left(E\left[\sum_{X\in\internal} \eps \sum C^*_X/\log n\right]\right)$ using (\ref{eq:Tbound}) (which depends on Assumption~\ref{assmp:events}).   Clearly the sum
of $C^*_X$ for $X$ ranging over nodes $X\in \internal$
in a particular level is at most $C(\pi_\root, V, W)$
(again using  Assumption~\ref{assmp:events} to assert
that the cost of $\improveout_X$ is less than the cost of $\pi_X$ at each node $X$).  By taking Assumption~\ref{assmp:qs} into account, $C(\pi_\root, V, W)$
is $O(C^*_\root)$.  Hence, summing over all $O(\log n)$ levels,
\begin{equation}
E\left [\sum_{X\in\internal} T_X\right  ] = O(\eps C^*_\root)\ .
\end{equation}
Let $\costout_X = C(\improveout_X, V_X, W_{|V_X})$ for all
$x\in \internal$.
Denote by $F$ the expression in the $O$-notation of the first summand in the RHS of (\ref{eq:difficult}), more precisely:
\begin{equation}\label{eq:Fbound} F = \sum_{X\in\internal} E\left [ \frac{\costout_X^{3/2}}{N_X}\right]\ , \end{equation}
where we remind the reader that $N_X=|V_X|$.
It will suffice to show that under Assumption~\ref{assmp:events}, the following inequality holds with probability $1$:
\begin{equation}\label{eq:Gbound}
G\left ((\costout_X)_{X\in \internal}, (N_X)_{X\in\internal}\right) := \sum_{X\in\internal} \costout_X^{3/2}/N_X \leq c_3\eps \costout_\root\ ,
\end{equation}
where $c_3>0$ is some global constant.
This turns out to require a bit of elementary calculus.  A complete proof of this assertion is not included in \cite{DBLP:conf/stoc/Kenyon-MathieuS07}, which
is an extened abstract.  We present a version of the proof here for the sake of completeness.

Under assumption~\ref{assmp:events}, the following two constraints
hold uniformly for all $X\in \internal$ with probability $1$:   Letting $C_X = C(\pi_X, V_X, W_{|V_X})$,
  \begin{enumerate}
\item[$(A1)$]
If $X$ is other than $\root$, let
$Y$ be its sibling and $P$ their parent.  In case $Y\in\internal$: \begin{equation}\label{eq:boundchildren}
\costout_X+\costout_Y \leq\costout_P\ .\end{equation}
   (In case $Y\in\leaf$, we simply have that $\costout_X \leq \costout_P$.\footnote{We can say something stronger in this case, but we won't need it here.})
To see this, notice that $\costout_X \leq C_X $, and similarly, in case $Y\in\internal$, $\costout_Y\leq C_Y$.  Clearly
$C_X+C_Y \leq \costout_P$, because $\pi_X, \pi_Y$ are
simply  restrictions of $\improveout_P$ to disjoint blocks of $V_P$.  The required inequality (\ref{eq:boundchildren}) is proven.
\item[$(A2)$] $\costout_X \leq c_2 \eps^2 N_X^2$ for some global $c_2>0$.  
\end{enumerate}

In order to show (\ref{eq:Gbound}), we may increase the values  $\costout_X$ for $X\neq \root$  in the following manner:  Start with the root node.
If it has no children, there is nothing to do because then $G=0$. If it has only  one child $X\in\internal$, continuously increase $\costout_X$ until
either $\costout_X = \costout_\root$ (making $(A1)$ tight) or $\costout_X = c_2\eps^2 N_X^2$ (making $(A2)$ above tight).  Then recurse on the subtree rooted by $X$.  In case $\root$ has two children
$X,Y\in\internal$ (say, $X$ on left), continuously increase $\costout_X$ until either $\costout_X+\costout_Y = \costout_\root$ ($(A1)$ tight) or until $\costout_X=c_2\eps^2 N_X^2$ ($(A2)$ tight) .  Then do the same for $\costout_Y$, namely, increase it until $(A1)$ is tight or until $\costout_Y = c_2\eps^2 N_Y^2$ ($(A2)$ tight).  Recursively perform
the same procedure for the subtrees rooted by $X,Y$.

After performing the above procedure, let $\internal_1$ denote the set of internal nodes $X$  for which $(A1)$ is tight, namely, either the sibling $Y$
of $X$ is a leaf and $\costout_X=\costout_P$ (where $P$ is $X$'s parent) or the sibling $Y\in\internal$ and $\costout_X+\costout_Y=\costout_P$ (in which case also $Y\in\internal_1$).  Let $\internal_2 = \internal \setminus \internal_1$.  By our construction, for all $X\in\internal_2$, $\costout_X = c_2\eps^2 N_X^2$.

Note that if $X\in\internal_2$ then its children (more precisely,  those in $\internal$) cannot be in $\internal_1$.  Indeed, this would violate $(A2)$ for at least one child, in virtue of the fact that $N_Y$ lies in the range $[N_X/3, 2N_X/3]$ for any child $Y$ of $X$.
Hence, the set $\internal_1\cup\{\root\}$ forms a connected subtree which we denote by $\tree_1$.  Let $P\in\tree_1$ be an internal node in $\tree_1$.
Assume it has one child in $\tree_1$, call it $X$.  Then $\costout_X = \costout_P$ and in virtue of $N_X\leq 2N_P/3$ we have
$\costout_P^{3/2}/N_P \leq (2/3)^{3/2} \costout_X^{3/2}/N_X$.  Now assume $P$ has two children $X,Y\in\tree_1$.  Then $\costout_X+\costout_Y=\costout_P$.  Using elementary calculus, we also have that $\costout_P^{3/2}/N_P \leq (\costout_X^{3/2}/N_X + \costout_Y^{3/2}/N_Y)/\sqrt{2}$ (indeed, the extreme case occurs for $N_X=N_Y=N_P/2$ and $\costout_X=\costout_Y = \costout_P/2$).  
We conclude that for any $P$ internal in $\tree_1$, the  corresponding contribution $\costout_P^{3/2}/N_P$ to $G$ is geometrically dominated by that
of its children in $\internal_1$.  Hence the entire sum $G_1 = \sum_{X\in\internal_1\cup\{\root\}} \costout_X^{3/2}/N_X$ is bounded by   $c_4\sum_{X\in\leaf_1} \costout_X^{3/2}/N_X$ for some constant $c_4$, where $\leaf_1$ is the set of leaves of $\tree_1$.
For each such leaf $X\in\leaf_1$, we have that $\costout_X^{3/2}/N_X \leq c_2^{3/2}\eps \costout_X$ (using $(A2)$), hence
$\sum_{X\in\leaf_1}\costout_X^{3/2}/N_X \leq \sum_{X\in\leaf_1} c_2^{3/2}\eps\costout_X \leq c_2^{3/2}\eps \costout_R$ (the
rightmost inequality in the chain follows from $\{V_X\}_{X\in\leaf_1}$ forming a disjoint cover of $V=V_\root$, together with $(A_1)$).  We conclude
that $G_1\leq c_4c_2^{3/2}\eps\costout_R$.

To conclude (\ref{eq:Gbound}), it remains to show that $G_2=G-G_1=\sum_{X\in\internal_2}$.  For $P\in G_2$, clearly $\costout_P^{3/2}/N_P = c_2^{3/2} \eps^3 N_P^3$.  Hence, if $X,Y\in G_2$ are children of $P$ in $\internal_2$ then $\costout_P^{3/2}/N_P \geq c_5 \costout_X^{3/2}/N_X + \costout_Y^{3/2}/N_Y$ and if  $X$ is the unique child of $P$ in $\internal_2$, then $\costout_P^{3/2}/N_P \geq c_5 \costout_X^{3/2}/N_X$, for some
global $c_5>1$.   In other words, the contribution to $G_2$ corresponding to $P$ geometrically dominates the sum of the corresponding contributions of its children.
We conclude that $G_2$ is at most some constant $c_6$ times $\sum_{X\in \operatorname{root}(\internal_2)} \costout_X^{3/2}/N_X$, where $\operatorname{root}(\internal_2)$ is the set of roots of the forrest induced by $\internal_2$.  As before, it is clear that  $\{V_X\}_{X\in\operatorname{root}(\internal_2)}$ is a disjoint collection, hence as before we conclude that $G_2 \leq c_7 \eps\costout_R$ for some global $c_7>0$.  
The assertion (\ref{eq:Gbound}) follows, and hence  (\ref{eq:Fbound}).

\noindent
Plugging our bounds in (\ref{eq:difficult}), we conclude that

$$ E\left [ \sum_{X\in \internal} \beta_X\right ] \leq C^*_\root(1+O(\eps)) - \sum_{X\in \leaf} E[C^*_X]\ .$$
Clearly $C(\hat \pi, V, W) = \sum_{X\in\internal} \beta_X + \sum_{X\in\leaf} C^*_X$. 
Hence $E[C(\hat \pi, V, W)] = (1+O(\eps))C_\root^* = (1+O(\eps))C^*$.
 We conclude the desired assertion on expectation.
%


\noindent
A simple counting of accesses to $W$ proves Theorem~\ref{thm:main}.

\section{Using Our Decomposition as a Preconditioner for SVM}\label{sec:svm}

We consider the following practical scenario, which is can be viewed as an  improvement over a version of the  well known SVMrank \cite{Joachims:2002:OSE:775047.775067,Peters:1994:FMA} for the preference label scenario. 

Consider the setting developed in Section~\ref{sec:learning}, where each element $u$ in $V$ is endowed with
a feature vector $\myphi(u)\in \R^d$ for some $d$ (we can also use infinite dimensiona spaces via kernels, but the
effective dimension is never more than $n=|V|$).  Assume, additionally, that $\|\phi(u)\|_2\leq 1$ for all $u\in V$ (otherwise,
normalize).
 Our hypothesis class $\H$ is parametrized by a weight vector $w\in \R^d$,
and each associated permutation $\pi_w$ is obtained by sorting the elements of $V$ in decreasing order of a score
given by $\score_w(u) = \langle \myphi(u), w \rangle$.  In other words, $u \prec_{\pi_w} v$ if $\score_w(u) > \score_w(v)$
(in case of ties, assume any arbitrary tie breaking scheme).

The following SVM formulaion is a convex relaxation for the problem of optimizing $C(h, V, W)$ over our chosen concept class $\H$:
\begin{eqnarray*}
\mbox{(SVM1)\ \ \ \ \ \ \ \ \ \ \ \ \ \ \ \ \ \ \ \ \ \ \ \  minimize }  & &F_1(w,\xi) =\sum_{u,v} \xi_{u,v} \\
\mbox{s.t. } \forall u,v: W(u,v)=1      & & \score_w(u)-\score_w(v) \geq 1-\xi_{u,v} \\
\forall u,v & & \xi_{u,v} \geq 0 \\
& & \|w\| \leq c \\
\end{eqnarray*}

\noindent
Instead of optimizing (SVM1) directly, we make the following
observation.  An $\eps$-good decomposition $V_1,\dots, V_k$
gives rise to a surrogate learning problem over $\Pi(V_1,\dots, V_k)\subseteq \Pi(V)$, such that optimizing over the restricted set does not compromise optimality over $\Pi(V)$ by more
than a relative  regret of $\eps$ (property (\ref{eq:globalorder})).  In turn, optimizing
over $\Pi(V_1,\dots, V_k)$ can be done separately for
each block $V_i$.  A natural underlying SVM corresponding
to this idea is captured as follows:
\begin{eqnarray*}
\mbox{(SVM2)\ \ \ \ \ \ \ \ \ \ \ \ \ \ \ \ \ \ \ \ \ \ \ \  minimize }  & &F_2(w,\xi) = \sum_{u,v\in \Delta_1\cup \Delta_2} \xi_{u,v} \\
\mbox{s.t. } \forall (u,v)\in \Delta_1\cup\Delta_2    & & \score_w(u)-\score_w(v) \geq 1-\xi_{u,v} \\
\forall u,v & & \xi_{u,v} \geq 0 \\
& & \|w\| \leq c\ , \\
\end{eqnarray*}

where $\Delta_1 = \bigcup_{1\leq i<j\leq k} V_i \times V_j$ and $\Delta_2 = \bigcup_{i=1}^k\{(u,v): u,v\in V_i \wedge W(u,v)=1\}$.

Abusing notation, for $w\in \R^d$ s.t. $\|w\|\leq c$, let $F_1(w)$ denote $\min F_1(w, \xi)$, where the minimum is taken over all
$\xi$ that satisfy the constraints of SVM1.   Observe that
$F_1(w)$ is simply $F_1(w, \xi)$, where $\xi$ is taken as:
\begin{equation}\label{eq:xi1}\xi_{u,v} = \begin{cases} \max\{0, 1-\score_w(u)+\score_w(v)\} & W(u,v)=1 \\ 0 & \mbox{otherwise} \end{cases} \ . \end{equation}

\noindent
Similarly define $F_2(w)$ as the minimizer of $F_2(w,\xi)$,
which is obtained by setting:
\begin{equation}\label{eq:xi2}\xi_{u,v} = \begin{cases} \max\{0, 1-\score_w(u)+\score_w(v)\} & (u,v)\in \Delta_1\cup \Delta_2 \\ 0 & \mbox{otherwise} \end{cases} \ . \end{equation}


\noindent
Let  $\pi^*$ denote the optimal solution to MFAST on $V,W$.

%
We do not know how to directly relate the optimal solution
to SVM1 and that of SVM2.  However, we can
 showwe can replace SVM2 with   a careful
sampling of constraints thereof, such that (i) the solution
to the subsampled SVM is optimal to within a relative
error of $\eps$ as a solution to SVM2, and (ii) the sampling
is such that only $O(n\polylog(n,\eps^{-1}))$ queries
to $W$ are necessary  in order to construct it.   
This result, which we quantify in what
follows, strongly relies on the local chaos property of the $\eps$-good decomposition (\ref{eq:localchaos}) and some combinatorics on permutations.


Our subsampled SVM which we denote by SVM3, is obtained as follows.  For ease of notation we assume that all blocks $V_1,\dots, V_k$ are big in $V$, otherwise a simple accounting of small blocks needs to be taken care of, adding notational clutter.  Let $\Delta_3$ be  a subsample of size $M$ (chosen shortly) of $\Delta_2$, each element chosen
uniformly at random from $\Delta_2$ (with repetitions - hence $\Delta_3$ is a multi-set).  Define:
\begin{eqnarray*}
\mbox{(SVM3)\ \ \ \ \ \ \ \ \ \ \ \ \ \ \ \ \ \ \ \ \ \ \ \  minimize }  & &F_3(w,\xi) = \sum_{u,v \in \Delta_1} \xi_{u,v} + \frac {{\sum_{i=1}^k {n_i\choose 2}}}{M}\sum_{u,v \in \Delta_3} \xi_{u,v}\\
\mbox{s.t. } \forall (u,v)\in \Delta_1\cup \Delta_3     & & \score_w(u)-\score_w(v) \geq 1-\xi_{u,v} \\
\forall u,v & & \xi_{u,v} \geq 0 \\
& & \|w\| \leq c \\
\end{eqnarray*}

As before, define $F_3(w)$ to be $F_3(w, \xi)$, where $\xi=\xi(w)$ is the minimizer of $F_3(w,\cdot)$ and is taken as
\begin{equation}\label{eq:xi3}\xi_{u,v} = \begin{cases} \max\{0, 1-\score_w(u)+\score_w(v)\} & (u,v)\in \Delta_1\cup \Delta_3 \\ 0 & \mbox{otherwise} \end{cases} \ . \end{equation}

 Our ultimate goal is to show
that for  quite  small $M$, SVM3 is a good approximation of SVM2.   To that end we first need another lemma.

\begin{lem}\label{lem:highcost}
Any feasible solution $(w, \xi)$ for SVM1 satisfies $\sum_{u,v} \xi_{u,v} \geq C(\pi^*, V, W)$.
\end{lem}
\begin{proof}
The following has been proven in  \cite{DBLP:journals/jacm/AilonCN08}:
Consider  \emph{non-transitive} triangles induced by $W$:  These are triplets $(u,v,y)$ of elements in $V$
such that $W(u,v)=W(v,y)=W(y,u)=1$.  Note that any permutation must disagree with at least one pair of elements contained
in a non-transitive triangle. Let $T$ denote the set of non-transitive triangles.  Now consider an assignment
of non-negative weights $\beta_t$ for each $t\in T$.  We say that the weight system $\{\beta_t\}_{t\in T}$ \emph{packs} $T$
 if for all $u,v\in V$ such that $W(u,v)=1$, the sum  $\sum_{(u,v) \mbox{ in } t} \beta_t$ is at most $1$.  (By \emph{$u,v$ in $t$} we mean that $u,v$ are two of the three elements inducing $t$.)  Let $\{\beta^*_t\}_{t\in T}$ be a weight system packing
$T$ with the maximum possible value of the sum of weights.  Then \begin{equation}\label{eq:subbeta} \sum_{t\in T} \beta^*_t\geq C(\pi^*, V, W)/3\ .\end{equation}

Now consider one non-transitive triangle $t=(u,v,y)\in T$.  We lower bound $\xi_{u,v}+\xi_{v,y}+\xi_{y,u}$
for any $\xi$ such that $w,\xi$ is a feasible solution to SVM1.  Letting $a=\score_w(u)-\score_w(v), b=\score_w(v)-\score_w(y), 
c=\score_w(y)-\score_w(u)$, we get from the constraints in SVM1 that $\xi_{u,v} \geq 1-a, \xi_{v,y}\geq 1-b, \xi_{y,u}\geq 1-c$.  But clearly $a+b+c=0$, hence \begin{equation}\label{eq:3lowerbound}\xi_{u,v}+\xi_{v,y}+\xi_{y,u}\geq 3\ .\end{equation}   Now notice that the objective function of SVM1 can be bounded from below as follows:
\begin{eqnarray*}
\sum_{u,v} \xi_{u,v} &\geq& \sum_{t=(u,v,y)\in T} \beta^*_t (\xi_{u,v}+\xi_{v,y}+\xi_{y,u}) \\
&\geq& \sum_{t=(u,v,y)\in T} \beta^*_t \cdot 3 \\
&\geq& C(\pi^*, V, W)\ .
\end{eqnarray*}
(The first inequality was due to the fact that $\{\beta^*_t\}_{t\in T}$ is a packing of the non-transitive triangles,
hence the total weight corresponding to each pair $u,v$ is at most $1$.  The second inequality is from (\ref{eq:3lowerbound}) and
the third is from (\ref{eq:subbeta}).)  This concludes the proof.
\end{proof}

\begin{thm}\label{thm:svmsample}
Let $\eps \in (0,1)$ and $M=O(\eps^{-6}(1+2c)^2d \log(1/\eps))$.  Then with
high constant probability, for all $w$ such that $\|w\|\leq c$,
$$ |F_3(w) - F_2(w)| = O(\eps F_2(w))\ .$$
\end{thm}
\begin{proof}
Let $B_d(c) = \{z\in \R^d: \|z\|\leq c\}$.
Fix a vector $w\in B_d(c)$.
%
 Over the random choice of $\Delta_3$, it is clear that $E[F_3(w)] = F_2(w)$.
We need a strong concentration bound.  From the observation that $|\xi_{u,v}|\leq 1+2c$ for all $u,v$, we conclude
(using Hoeffding bound) that for all $\mu>0$,
\begin{equation}\label{eq:hoeffding2}
\Pr[|F_3(w) - F_2(w)| \geq \mu] \leq \exp\left \{\frac {-\mu^2M}{\left(\sum_{i=1}^k{n_i\choose 2}(1+2c)\right)^2}\right \}\ .
\end{equation}
Let $\eta=\eps^3$ and  consider an $\eta$-net of vectors $w$ in the ball $B_d(c)$.  By this we mean a subset $\Gamma\subseteq B_d(c)$ such that for all $z\in B_d(c)$ there
exists $w\in \Gamma$ s.t. $\|z-w\|\leq \eta$.  Standard volumetric arguments imply that there exists such a set $\Gamma$ of cardinality at most $(c/\eta)^d$.

Let $z\in \Gamma$ and $w\in B_d(c)$ such that $\|w-z\|\leq \eta$.  From the definition of $F_2,F_3$,
it is clear that 
\begin{equation}\label{eq:epsnet}|F_2(w) - F_2(z)| \leq \sum_{i=1}^k{n_i\choose 2}\eps^3 ,\ \ \ |F_3(w)-F_3(z)| \leq \sum_{i=1}^k{n_i\choose 2}\eps^3 \ .\end{equation}

Using (\ref{eq:hoeffding2}), we conclude that for any $\mu>0$, by taking $M=
O(\mu^{-2}(\sum {n_i\choose 2})^2(1+2c)^2d \log(c\eta^{-1}))$,
with constant probability over the choice of $\Delta_3$,
uniformly for all $z\in \Gamma$:
$$ |F_3(z) - F_2(z)| \leq \mu\ .$$
Take $\mu = \eps^3 \sum_{i=1}^k {n_i\choose 2}$.  We conclude (plugging in our choice of $\mu$ and the definition of $\eta$) that by choosing $$M=O(\eps^{-6}(1+2c)^2d\log(c/\eps))\ ,$$ with constant probability, uniformly for all $z\in \Gamma$:
\begin{equation}\label{eq:abc} |F_3(z) - F_2(z)| \leq \eps^3\sum_{i=1}^k{n_i\choose 2}\ .\end{equation}
Using (\ref{eq:epsnet}) and the triangle inequality, we conclude that for all $w\in B_d(c)$,
\begin{equation}\label{eq:abc} |F_3(w) - F_2(w)| \leq 3 \eps^3\sum_{i=1}^k{n_i\choose 2}\ .\end{equation}
By property (\ref{eq:localchaos}) of the $\eps$-goodness
definition, (\ref{eq:abc}) imples  $$|F_3(w)-F_2(w)| \leq 3\eps \min_{\pi\in \Pi(V)}\sum_{i=1}^k C(\pi_{|V_i}, V_i, W_{|V_i}) = 3\eps \sum_{i=1}^k \min_{\sigma\in \Pi(V_i)}C(\sigma, V_i, W_{|V_i})\ .$$
By Lemma~\ref{lem:highcost} applied separately in each block $V_i$, this implies
$$ |F_3(w)-F_2(w)| \leq 3\eps \sum_{i=1}^k \sum_{u,v\in V_i} \xi_{u,v} = 3\eps F_2(w)  ,$$
(where $\xi=\xi(w)$ is as defined in (\ref{eq:xi2}).)
This concludes the proof.

\end{proof}

\section*{Acknowledgements}
The author gratefully acknowledges the help of Warren Schudy with derivation of some of the bounds in this work.

\bibliographystyle{amsplain}
\bibliography{active_ranking}

\begin{appendix}

\section{Linear VC Bound of Permutation Set}\label{appendix:vcbound}
To see why the VC dimension of the set of permutations viewed as binary function over the set of all possible ${n\choose 2}$ preferences, it is enough
to show that any collection of $n$ pairs of   elements cannot be \emph{shattered} by the set of permutation. (Refer to the definition of VC dimension \cite{vapnik:264}
for a definition of shattering).
 Indeed, any such collection must  contain a cycle, and the set of permutations cannot direct a cycle cyclically.
\end{appendix}
\end{document}